\documentclass[submission,copyright,creativecommons]{eptcs}
\providecommand{\event}{GANDALF 2020} 
\usepackage{breakurl}             
\usepackage{underscore}           

\usepackage{times}
\usepackage{amsmath}
\usepackage{amsthm} 
\usepackage{amssymb }
\usepackage{caption}
\usepackage{tikz}

\newtheorem{theorem}{Theorem}[section]

\newtheorem{lemma}[theorem]{Lemma}

\newtheorem{proposition}[theorem]{Proposition}

\newtheorem{example}{Example}[section]

\DeclareMathAlphabet{\altmathcal}{OMS}{cmsy}{m}{n}
\newcommand\stam[1]{}

\newcommand{\A}{\altmathcal {A}}
\newcommand{\B}{\altmathcal {B}}
\newcommand{\C}{\altmathcal {C}}
\newcommand{\D}{\altmathcal {D}}
\newcommand{\E}{\altmathcal {E}}
\renewcommand{\S}{\altmathcal {S}}
\newcommand{\T}{\altmathcal {T}}

\newcommand{\zug}[1]{\langle #1 \rangle}

\title{Canonicity in GFG and Transition-Based Automata}
\author{Bader Abu Radi and Orna Kupferman
\institute{School of Computer Science and Engineering, The Hebrew University, Israel}
\email{bader.aburadi@gmail.com \quad orna@cs.huji.ac.il}
}
\def\titlerunning{Canonicity in GFG and Transition-Based Automata}
\def\authorrunning{B. Abu Radi and O. Kupferman}
\begin{document}
\maketitle

\begin{abstract}
	
	Minimization of deterministic automata on finite words results in a {\em canonical\/} automaton. For deterministic automata on infinite words, no canonical minimal automaton exists, and a language may have different minimal deterministic B\"uchi (DBW) or co-B\"uchi (DCW) automata.

	In recent years, researchers have studied  {\em good-for-games\/} (GFG) automata -- nondeterministic automata that can resolve their nondeterministic choices in a way that only depends on the past. Several applications of automata in formal methods, most notably synthesis, that are traditionally based on deterministic automata, can instead be based on GFG automata.

	The {\em minimization\/} problem for DBW and DCW is NP-complete, and it stays NP-complete for GFG B\"uchi and co-B\"uchi automata. On the other hand, minimization of GFG co-B\"uchi automata with {\em transition-based\/} acceptance (GFG-tNCWs) can be solved in polynomial time. In these automata, acceptance is defined by a set $\alpha$ of transitions, and a run is accepting if it traverses transitions in $\alpha$ only finitely often. This raises the question of canonicity of minimal deterministic and GFG automata with transition-based acceptance.

	In this paper we study this problem. We start with GFG-tNCWs and show that the safe components (that is, these obtained by restricting the transitions to these not in $\alpha$) of all minimal GFG-tNCWs are isomorphic, and that by saturating the automaton with transitions in $\alpha$ we get isomorphism among all minimal GFG-tNCWs. Thus, a canonical form for  minimal GFG-tNCWs can be obtained in polynomial time. We continue to DCWs with transition-based acceptance (tDCWs), and their dual tDBWs. We show that here, while no canonical form for minimal automata exists, restricting attention to the safe components is useful, and implies that the only minimal tDCWs that have no canonical form are these for which the transition to the GFG model results in strictly smaller automaton, which do have a canonical minimal form. 
	
\end{abstract}

\section{Introduction}
\label{intro}

Automata theory is one of the longest established areas in computer science. A classical problem in automata theory is {\em minimization}: generation of an equivalent automaton with a minimal number of states. For deterministic automata on finite words, a minimization algorithm, based on the Myhill-Nerode right congruence \cite{Myh57,Ner58}, generates in polynomial time a canonical minimal deterministic automaton \cite{Hop71}. Essentially, the canonical automaton, a.k.a. the {\em quotient automaton}, is obtained by merging equivalent states.

A prime application of automata theory is specification, verification, and synthesis of reactive systems  \cite{VW94,Kup15}. Since we care about the on-going behaviors of nonterminating  systems, the automata run on infinite words and define $\omega$-regular languages. Acceptance in such automata is determined according to the set of states that are visited infinitely often during the run. In B\"uchi automata \cite{Buc62} (NBW and DBW, for  nondeterministic and deterministic B\"uchi word automata, respectively), the acceptance condition is a subset $\alpha$ of states, and a run is accepting iff it visits $\alpha$ infinitely often. Dually, in co-B\"uchi automata (NCW and DCW), a run is accepting iff it visits $\alpha$ only finitely often.

For $\omega$-regular languages, no canonical minimal deterministic automaton exists, and a language may have different minimal DBWs or DCWs. 
Consider for example the DCWs $\A_1$ and $\A_2$ appearing in Figure~\ref{2min dcws}. Both are minimal DCWs for the language $L=(a+b)^* \cdot (a^\omega + b^\omega)$ (``only finitely many $a$'s or only finitely many $b$'s''; it is easier to see this by considering the dual DBWs, for ``infinitely many $a$'s and infinitely many $b$'s"). 

\begin{figure}[htb]	
	
	\begin{center}
		
		\includegraphics[width=.8\textwidth]{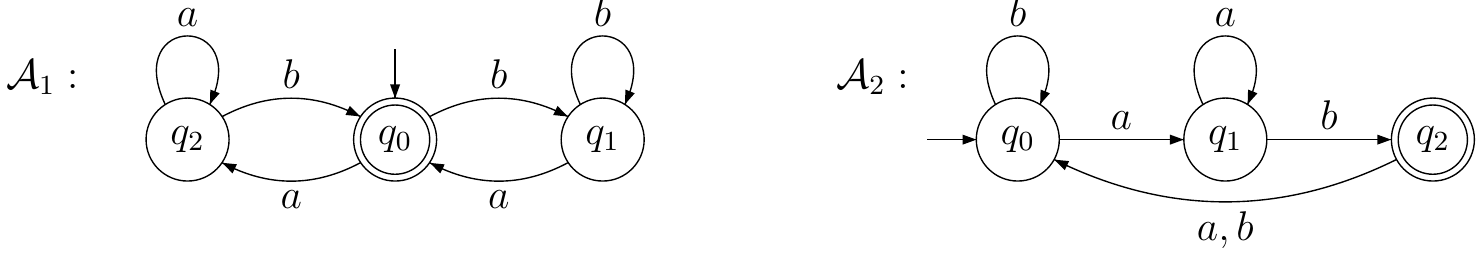}
		
		\caption{The DCWs $\A_1$ and $\A_2$.}
		
		\label{2min dcws}
		
	\end{center}
	
\end{figure}

	\vspace{-6mm}
Since all the states of $\A_1$ and $\A_2$ recognize the language $L$ and may serve as initial states, Figure~\ref{2min dcws} actually presents six different DCWs for $L$, and more three-state DCWs for $L$ exist. The DCWs $\A_1$ and $\A_2$, are however ``more different" than variants of $\A_1$ obtained by changing the initial state: they have a different structure, or more formally, there is no isomorphism between their graphs.

In some applications of automata on infinite words, such as model checking, algorithms can proceed with nondeterministic automata. In other applications, such as synthesis and control, they cannot. The algorithms for these applications involve solving a game that is played on an arena that is based on the automaton. The difficulty in using nondeterministic automata in such game-based algorithms lies in the fact that when a player resolves nondeterminism, her choices should accommodate all possible futures. 

A study of nondeterministic automata that can resolve their nondeterministic choices in a way that only depends on the past started in \cite{KSV06}, where the setting is modeled by means of tree automata for derived languages. It then continued by means of  {\em good for games\/}  (GFG) automata \cite{HP06}.\footnote{GFGness is also used in \cite{Col09} in the framework of cost functions under the name ``history-determinism''.} A nondeterministic automaton $\A$ over an alphabet $\Sigma$ is GFG if there is a strategy $g$ that maps each finite word $u \in \Sigma^*$ to the transition to be taken after $u$ is read; and following $g$ results in accepting all the words in the language of $\A$. Note that a state $q$ of $\A$ may be reachable via different words, and $g$ may suggest different transitions from $q$ after different words are read. Still, $g$ depends only on the past, namely on the word read so far. Obviously, there exist GFG automata: deterministic ones, or nondeterministic ones that are {\em determinizable by pruning\/} (DBP); that is, ones that just add transitions on top of a deterministic automaton. In fact, the GFG automata constructed in~\cite{HP06} are DBP.\footnote{As explained in~\cite{HP06}, the fact that the GFG automata constructed there are DBP does not contradict their usefulness in practice, as their transition relation is simpler than the one of the embodied deterministic automaton and it can be defined symbolically.}

In terms of expressive power, it is shown in~\cite{KSV06,NW98} that GFG automata with an acceptance condition $\gamma$ (e.g., B\"uchi) are as expressive as deterministic $\gamma$ automata. The picture in terms of succinctness is diverse. For automata on finite words, GFG automata are always DBP~\cite{KSV06,Mor03}. For automata on infinite words, in particular NBWs and NCWs, GFG automata need not be DBP~\cite{BKKS13}. 
Moreover, the best known determinization construction for GFG-NBWs is quadratic, whereas determinization of GFG-NCWs has an exponential blow-up lower bound~\cite{KS15}. Thus, GFG automata on infinite words are more succinct (possibly even exponentially) than deterministic ones.\footnote{We note that some of the succinctness results are known only for GFG automata with \emph{transition-based} acceptance.} 
Further research studies characterization, typeness, complementation, and further constructions and decision procedures for GFG automata \cite{KS15,BKS17,BK18}, as well as an extension of the GFG setting to pushdown $\omega$-automata \cite{LZ20}.

Recall that for automata on finite words, a minimal deterministic automaton can be obtained by merging equivalent states. For general DBWs (and hence, also DCWs, as the two dualize each other), merging equivalent states fails, and minimization is NP-complete~\cite{Sch10}.
Proving NP-hardness, Schewe used a reduction from the vertex-cover problem \cite{Sch10}. Essentially, the choice of a vertex cover in a given graph $G$ is reduced to a choice of a set of states that should be duplicated in a DBW induced by $G$. The duplication is needed for the definition of the acceptance condition, and is not needed when the DBW is defined with a {\em transition-based\/} acceptance condition.  In such automata, the acceptance condition is given by a subset $\alpha$ of the transitions, and a run is required to traverse transitions in $\alpha$ infinitely often (in B\"uchi automata, denoted tNBW), or finitely often (in co-B\"uchi automata, denoted tNCW). Thus, while minimization is NP-complete for DBW and DCW, its complexity is open for tDBWs and tDCWs.
Beyond the theoretical interest, there is recently growing use of transition-based automata in practical applications, with evidences they offer a simpler translation of LTL formulas to  automata and enable simpler constructions and decision procedures  \cite{GL02,DLFMRX16,SEJK16,LKH17}.

In \cite{AK19}, we described a polynomial-time algorithm for the minimization of GFG-tNCWs. Consider a GFG-tNCW $\A$. Our algorithm is based on an analysis of the {\em safe components\/} of $\A$, namely its strongly connected components obtained by removing transitions in $\alpha$. Note that every accepting run of $\A$ eventually reaches and stays forever in a safe component. We showed that a minimal GFG-tNCW equivalent to $\A$ can be obtained by defining an order on the safe components, and applying the quotient construction on a GFG-tNCW obtained by restricting attention to states that belong to components that form a frontier in this order. 
Considering GFG-tNCWs rather than DBWs involves two modifications of the original question: a transition to GFG rather than deterministic automata, and a transition to transition-based rather than state-based acceptance. A natural question that arises is whether both modifications are crucial for efficiency. 
It was shown recently \cite{Sch20} that the NP-completeness proof of Schewe for DBW minimization can be generalized to GFG-NBWs and GFG-NCWs.
This suggests that the consideration of transition-based acceptance has been crucial, and makes the study of tDBW and tDCW very appealing.

Minimization and its complexity are tightly related to the canonicity question. Recall that $\omega$-regular languages do not have a unique minimal DBW or DCW. In this paper we study canonicity for GFG and transition-based automata. We start with GFG-tNCWs and show that all minimal GFG-tNCWs are {\em safe isomorphic}, namely their safe components are isomorphic\footnote{In our results, we assume the GFG-tNCWs are {\em nice}: they satisfy some syntactic and semantic properties that can be easily obtained from every GFG-tNCW.}. More formally, if $\A_1$ and $\A_2$ are minimal GFG-tNCWs for the same language, then there exists a bijection between the state spaces of $\A_1$ and $\A_2$ that induces a bijection between their $\bar{\alpha}$-transitions (these not in $\alpha$).  We then show that by saturating the GFG-tNCW with $\alpha$-transitions we get isomorphism among all minimal automata. 
We suggest two possible saturations.
One adds as many $\alpha$-transitions as possible, and the second does so in a way that preserves $\alpha$-homogeneity, thus for every state $q$ and letter $\sigma$, all the transitions labeled $\sigma$ from $q$ are $\alpha$-transitions or are all 
$\bar{\alpha}$-transitions. Since the minimization algorithm of \cite{AK19} generates minimal $\alpha$-homogenous GFG-tNCWs, it follows that both forms of canonical minimal GFG-tNCW can be obtained in polynomial time. 

We then show that, as has been the case with minimization, GFGness is not a sufficient condition for canonicity, raising the question of canonicity in tDCWs. Note that unlike the GFG-tNCW setting, here dualization of the acceptance condition complements the language of an automaton, and thus our results apply also to canonicity of tDBWs. We start with some bad news, showing that as has been the case with DCWs and DBW, minimal tDCWs and tDBWs need not be isomorphic. Moreover, being deterministic, we cannot saturate their transitions and make them isomorphic. On the positive side, safe isomorphism is helpful also in the tDCW setting: Consider an $\omega$-regular language $L$. Recall that the minimal  GFG-tNCW for $L$ may be smaller than a minimal tDCW for $L$ \cite{KS15}. We say that $L$ is {\em tDCW-positive\/} if this is not the case. We prove that all minimal tDCWs for a tDCW-positive $\omega$-regular language are safe isomorphic. Note that for such languages, we also know how to generate a minimal tDCW in polynomial time. For $\omega$-regular languages that are not tDCW-positive, safe isomorphism is left open. For such languages, however, we care more about minimal GFG-tNCWs, which do have a canonical form.  Also, all natural $\omega$-regular languages are tDCW-positive, and in fact the existence of $\omega$-regular languages that are not tDCW-positive has been open for quite a while \cite{BKKS13}. Accordingly, we view our results as good news about canonicity in deterministic automata with transition-based acceptance. 


\section{Preliminaries}
\label{prelim}

For a finite nonempty alphabet $\Sigma$, an infinite {\em word\/} $w = \sigma_1 \cdot \sigma_2 \cdots \in \Sigma^\omega$ is an infinite sequence of letters from $\Sigma$. 
A {\em language\/} $L\subseteq \Sigma^\omega$ is a set of words. We denote the empty word by $\epsilon$, and the set of finite words over $\Sigma$ by $\Sigma^*$. For $i\geq 0$, we use $w[1, i]$ to denote the (possibly empty) prefix $\sigma_1\cdot \sigma_2 \cdots  \sigma_i$ of $w$ and use $w[i+1, \infty]$ to denote its suffix $\sigma_{i+1} \cdot  \sigma_{i+2} \cdots$.

A \emph{nondeterministic automaton} over infinite words is $\A = \langle \Sigma, Q, q_0, \delta, \alpha  \rangle$, where $\Sigma$ is an alphabet, $Q$ is a finite set of \emph{states}, $q_0\in Q$ is an \emph{initial state}, $\delta: Q\times \Sigma \to 2^Q\setminus \emptyset$ is a \emph{transition function}, and $\alpha$ is an \emph{acceptance condition}, to be defined below. For states $q$ and $s$ and a letter $\sigma \in \Sigma$, we say that $s$ is a $\sigma$-successor of $q$ if $s \in \delta(q,\sigma)$.  
The \emph{size} of $\A$, denoted $|\A|$, is defined as its number of states, thus, $|\A| = |Q|$. 
Note that $\A$ is {\em total}, in the sense that it has at least one successor for each state and letter, and that $\A$ may be \emph{nondeterministic}, as the transition function may specify several successors for each state and letter.
If $|\delta(q, \sigma)| = 1$ for every state $q\in Q$ and letter $\sigma \in \Sigma$, then $\A$ is \emph{deterministic}.

When $\A$ runs on an input word, it starts in the initial state and proceeds according to the transition function. Formally, a \emph{run}  of $\A$ on $w = \sigma_1 \cdot \sigma_2 \cdots \in \Sigma^\omega$ is an infinite sequence of states $r = r_0,r_1,r_2,\ldots \in Q^\omega$, such that $r_0 = q_0$, and for all $i \geq 0$, we have that $r_{i+1} \in \delta(r_i, \sigma_{i+1})$. We sometimes extend $\delta$ to sets of states and finite words. Then, $\delta: 2^Q\times \Sigma^* \to 2^Q$ is such that for every $S \in 2^Q$, finite word $u\in \Sigma^*$, and letter $\sigma\in \Sigma$, we have that $\delta(S, \epsilon) = S$, $\delta(S, \sigma) = \bigcup_{s\in S}\delta(s, \sigma)$, and $\delta(S, u \cdot \sigma) = \delta(\delta(S, u), \sigma)$. Thus, $\delta(S, u)$ is the set of states that $\A$ may reach when it reads $u$ from some state in $S$. 

The transition function $\delta$ induces a transition relation $\Delta \subseteq Q\times \Sigma \times Q$, where for every two states $q,s\in Q$ and letter $\sigma\in \Sigma$, we have that $\langle q, \sigma, s \rangle \in \Delta$ iff $s\in \delta(q, \sigma)$. 
We sometimes view the run $r = r_0,r_1,r_2,\ldots$ on $w = \sigma_1 \cdot \sigma_2 \cdots$ as an infinite sequence of successive transitions $\zug{r_0,\sigma_1,r_1}, \zug{r_1,\sigma_2,r_2},\ldots \in \Delta^\omega$.
The acceptance condition $\alpha$ determines which runs are ``good''. We consider here \emph{transition-based} automata, in which $\alpha$ refers to the set of transitions that are traversed infinitely often during the run; specifically, $\alpha\subseteq \Delta$. We use the terms {\em $\alpha$-transitions\/} and  {\em $\bar{\alpha}$-transitions\/} to refer to  transitions in $\alpha$ and in $\Delta \setminus \alpha$, respectively. We also refer to restrictions $\delta^\alpha$ and  $\delta^{\bar{\alpha}}$ of $\delta$, where for all $q,s \in Q$ and $\sigma \in \Sigma$, we have that $s \in \delta^\alpha(q, \sigma)$ iff $\zug{q,\sigma,s} \in \alpha$, and $s \in \delta^{\bar{\alpha}}(q, \sigma)$ iff $\zug{q,\sigma,s} \in \Delta \setminus \alpha$.
For a run $r \in \Delta^\omega$, let ${\it inf}(r)\subseteq \Delta$ be the set of transitions that $r$ traverses infinitely often. Thus, 
${\it inf}(r) = \{  \langle q, \sigma, s\rangle \in \Delta: q = r_i, \sigma = \sigma_{i+1} \text{ and } s = r_{i+1} \text{ for infinitely many $i$'s}   \}$. In {\em co-B\"uchi\/} automata, a run $r$ is \emph{accepting} iff ${\it inf}(r)\cap \alpha = \emptyset$, thus if $r$ traverses transitions in $\alpha$ only finitely often. A run that is not accepting is \emph{rejecting}.  A word $w$ is accepted by $\A$ if there is an accepting run of $\A$ on $w$. The language of $\A$, denoted $L(\A)$, is the set of words that $\A$ accepts. Two automata are \emph{equivalent} if their languages are equivalent. We use tNCW and tDCW to abbreviate nondeterministic and deterministic transition-based co-B\"uchi automata over infinite words, respectively.

We continue to definitions and notations that are relevant to our study. See Section~\ref{app glos} for a glossary. 
For an automaton $\A = \langle \Sigma, Q, q_0, \delta, \alpha \rangle$, and a state $q\in Q$, we define $\A^q$ to be the automaton obtained from $\A$ by setting the initial state to be $q$. Thus, $\A^q = \langle \Sigma, Q, q, \delta, \alpha \rangle$.
We say that two states $q,s\in Q$ are \emph{equivalent}, denoted $q \sim_{\A} s$, if $L(\A^q) = L(\A^s)$. 
The automaton $\A$ is \emph{semantically deterministic} if different nondeterministic choices lead to equivalent states. Thus, for every state $q\in Q$ and letter $\sigma \in \Sigma$, all the $\sigma$-successors of $q$ are equivalent: for every two states $s, s'\in Q$ such that  $\langle q, \sigma, s\rangle$ and $\langle q, \sigma, s'\rangle$ are in $\Delta$, we have that $s \sim_{\A} s'$. 
The following proposition follows immediately from the definitions.

\begin{proposition}\label{pruned-corollary}
	Consider a semantically deterministic automaton $\A$, states $q,s \in Q$, letter $\sigma\in \Sigma$, and transitions $\langle q, \sigma, q'\rangle,\langle s, \sigma, s'\rangle \in \Delta$. If $q \sim_{\A} s$, then $q' \sim_{\A} s'$. 
\end{proposition}

An automaton $\A$ is \emph{good for games} (\emph{GFG}, for short) if its nondeterminism can be resolved based on the past, thus on the prefix of the input word read so far. Formally, $\A$ is \emph{GFG} if there exists a {\em strategy\/} $f:\Sigma^* \to Q$ such that the following holds: 
\begin{enumerate}
	\item 
	The strategy $f$ is consistent with the transition function. That is, for every finite word $u \in \Sigma^*$ and letter $\sigma \in \Sigma$, we have that $\zug{f(u),\sigma,f(u \cdot \sigma)} \in \Delta$. 
	\item
	Following $f$ causes $\A$ to accept all the words in its language. That is, for every infinite word $w = \sigma_1 \cdot \sigma_2 \cdots \in \Sigma^\omega$, if $w \in L(\A)$, then the run $f(w[1, 0]), f(w[1, 1]), f(w[1, 2]), \ldots$, which we denote by $f(w)$, is accepting. 
\end{enumerate}
We say that the strategy $f$ \emph{witnesses} $\A$'s GFGness. 
For an automaton $\A$, we say that a state $q$ of $\A$ is \emph{GFG} if $\A^q$ is GFG. 
Note that every deterministic automaton is GFG. We say that a GFG automaton $\A$ is {\em determinizable by prunning} (DBP) if we can remove some of the transitions of $\A$ and get a deterministic automaton that recognizes $L(\A)$.

Consider a directed graph $G = \langle V, E\rangle$. A \emph{strongly connected set\/} in $G$ (SCS, for short) is a set $C\subseteq V$ such that for every two vertices $v, v'\in C$, there is a path from $v$ to $v'$. A SCS is \emph{maximal} if it is maximal w.r.t containment, that is, for every non-empty set $C'\subseteq V\setminus C$, it holds that $C\cup C'$ is not a SCS. The \emph{maximal strongly connected sets} are also termed \emph{strongly connected components} (SCCs, for short). The \emph{SCC graph of $G$} is the graph defined over the SCCs of $G$, where there is an edge from a SCC $C$ to another SCC $C'$ iff there are two vertices $v\in C$ and $v'\in C'$ with $\langle v, v'\rangle\in E$. A SCC is \emph{ergodic} iff it has no outgoing edges in the SCC graph. 
The SCC graph of $G$ can be computed in linear time by standard SCC algorithms \cite{Tar72}. 

An automaton $\A = \langle \Sigma, Q, q_0, \delta, \alpha\rangle$ induces a directed graph $G_{\A} = \langle Q, E\rangle$, where $\langle q, q'\rangle\in E$ iff there is a letter $\sigma \in \Sigma$ such that $\langle q, \sigma, q'\rangle \in \Delta$. The SCSs and SCCs of $\A$ are those of $G_{\A}$. We say that a tNCW $\A$ is \emph{safe deterministic} if by removing its $\alpha$-transitions, we get a (possibly not total)  deterministic automaton. Thus, $\A$ is \emph{safe deterministic} if for every state $q\in Q$ and letter $\sigma\in \Sigma$, it holds that $|\delta^{\bar{\alpha}}(q, \sigma)|\leq 1$.
We refer to the SCCs we get by removing $\A$'s $\alpha$-transitions as the \emph{safe components} of $\A$; that is, the \emph{safe components} of $\A$ are the SCCs of the graph $G_{\A^{\bar{\alpha}}} = \langle Q, E^{\bar{\alpha}} \rangle$, where $\zug{q, q'}\in E^{\bar{\alpha}}$ iff there is a letter $\sigma\in \Sigma$ such that $q'\in \delta^{\bar{\alpha}}(q, \sigma)$. We denote the set of safe components of $\A$ by $\S(\A)$. For a safe component $S\in \S(\A)$, the \emph{size} of $S$, denoted $|S|$, is the number of states in $S$. 
Note that an accepting run of $\A$ eventually gets trapped in one of $\A$'s safe components. A tNCW $\A$ is \emph{normal} if 
there are no $\bar{\alpha}$-transitions connecting different safe components. That is,
for all states $q$ and $s$ of $\A$, if there is a path of $\bar{\alpha}$-transitions from $q$ to $s$, then there is also a path of $\bar{\alpha}$-transitions from $s$ to $q$.

We now combine several properties defined above and say that a GFG-tNCW $\A$ is {\em nice\/} if all the states in $\A$ are reachable and GFG, and $\A$ is normal, safe deterministic, and semantically deterministic. 
As Theorem~\ref{nice} below shows, each of these properties can be obtained in at most polynomial time, and without the properties being conflicting. 
\begin{theorem}\label{nice}{\rm \cite{KS15, AK19}}
	Every GFG-tNCW $\A$ can be turned, in polynomial time, into an equivalent nice GFG-tNCW $\B$ such that $|\B|\leq |\A|$.
\end{theorem}

Consider a tNCW $\A = \langle \Sigma, Q, q_0, \delta, \alpha \rangle$. 
A run $r$ of $\A$ is \emph{safe} if it does not traverse $\alpha$-transitions. 
The \emph{safe language} of $\A$, denoted $L_{safe}(\A)$, is the set of infinite words $w$, such that there is a safe run of $\A$ on $w$. 
Recall that two states $q,s\in Q$ are equivalent ($q \sim_{\A} s$) if $L(\A^q) = L(\A^s)$. 
Then, $q$ and $s$ are  \emph{strongly-equivalent}, denoted $q \approx _{\A} s$, if $q \sim_{\A} s$ and $L_{safe}(\A^q) = L_{safe}(\A^s)$. Finally, $q$ is \emph{subsafe-equivalent to} $s$, denoted $q\precsim_{\A} s$, if $q \sim_{\A} s$ and $L_{safe}(\A^q) \subseteq L_{safe}(\A^s)$. Note that the three relations are transitive. When $\A$ is clear from the context, we omit it from the notations, thus write $L_{safe}(q)$, $q\precsim s$, etc. 
The tNCW $\A$ is \emph{safe-minimal} if it has no strongly-equivalent states. Then, $\A$ is \emph{safe-centralized} if for every two states $q, s\in Q$, if $q \precsim s$, then $q$ and $s$ are in the same safe component of $\A$. Finally, $\A$ is \emph{$\alpha$-homogenous} if for every state $q\in Q$ and letter $\sigma \in \Sigma$, either $\delta^\alpha_\A(q, \sigma) =\emptyset$ or $\delta^{\bar{\alpha}}_\A(q, \sigma) = \emptyset$. Thus, either all the $\sigma$-labeled transitions from $q$ are $\alpha$-transitions, or they are all $\bar{\alpha}$-transitions.

\begin{example}
{\rm Consider the tDCW $\A$ appearing in Figure~\ref{safe example}. The dashed transitions are $\alpha$-transitions. 
		All the states of $\A$ are equivalent, yet they all differ in their safe language. Accordingly, $\A$ is safe-minimal.
		Since $a^\omega= L_{safe}(\A^{q_2}) \subseteq L_{safe}(\A^{q_0})$, we have that $q_2 \precsim q_0$. Hence, as $q_0$ and $q_2$ are in different safe components, the tDCW $\A$ is not safe-centralized. \hfill \qed}

\begin{figure}[htb]	
	\begin{center}
	\vspace{-4mm}
		\includegraphics[width=.3\textwidth]{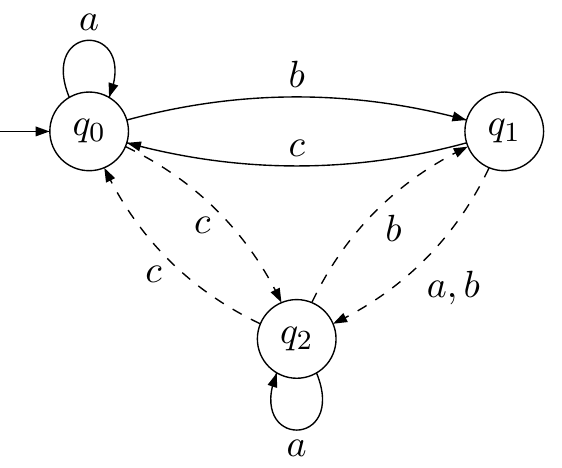}
		\caption{The tDCW $\A$.}
		\label{safe example}
	\end{center}
	\end{figure}
\end{example}

	\vspace{-6mm}
	
	The following properties of nice GFG-tNCWs are proven in \cite{AK19}.
\begin{proposition}\label{equiv-to-equiv}
	Consider a nice GFG-tNCW $\A$ and states $q$ and $s$ of $\A$ such that 
	$q \approx s$ ($q \precsim s$). For every letter $\sigma \in \Sigma$ and $\bar{\alpha}$-transition $\langle q, \sigma, q'\rangle$, there is an $\bar{\alpha}$-transition $\langle s, \sigma, s'\rangle$ such that $q' \approx s'$ ($q' \precsim s'$, respectively).
\end{proposition}

\begin{proposition}\label{there is the same safe}
	Let $\A$ and $\B$ be equivalent nice GFG-tNCWs. For every state $p$ of $\A$, there are states $q$ of $\A$ and $s$ of $\B$, such that $p \precsim q$ and $q \approx s$.
\end{proposition}

\begin{lemma}
	\label{injection}
	Consider a nice GFG-tNCW $\A$. If $\A$ is safe-centralized and safe-minimal, then for every nice GFG-tNCW $\B$ equivalent to $\A$, there is an injection $\eta: \S(\A) \to \S(\B)$ such that for every safe component $T\in \S(\A)$, it holds that $|T|\leq |\eta(T)|$.
\end{lemma}

\section{Minimizing GFG-tNCW}\label{sec ak19}

A GFG-tNCW $\A$ is \emph{minimal} if for every equivalent GFG-tNCW $\B$, it holds that $|\A|\leq |\B|$.
In this section, we review the minimization construction of \cite{AK19}, highlighting its properties that are important for the canonization results. The algorithm is based on the following theorem.

\begin{theorem}\label{C is minimal} 
	Consider a nice GFG-tNCW $\A$. If $\A$ is safe-centralized and safe-minimal, then $\A$ is a minimal  GFG-tNCW for $L(\A)$. 
\end{theorem}

Thus, minimization involves two steps: safe centralization and safe minimization. 

\paragraph{Step 1: Safe centralization}
Consider a nice GFG-tNCW $\A = \langle \Sigma, Q_{\A}, q^0_{\A}, \delta_{\A}, \alpha_{\A}\rangle$. Recall that  $\S(\A)$ denotes the set of safe components of $\A$. 
Let $H\subseteq \S(\A)\times \S(\A)$ be such that for all safe components $S, {S}'\in \S(\A)$, we have that $H(S, {S}')$ iff there exist states $q\in S$ and $q' \in {S}'$ such that $q\precsim q'$.  
The relation $H$ is transitive: for every safe components $S, {S}', {S}'' \in \S(\A)$, if $H(S, {S}')$ and $H({S}', {S}'')$, then $H(S, {S}'')$.
We say that a set ${\S}\subseteq \S(\A)$ is a \emph{frontier of $\A$\/} if for every safe component $S\in \S(\A)$, there is a safe component ${S}'\in {\S}$ with $H(S, {S}')$, and for all safe components $S, {S}' \in {\S}$ such that $S\neq {S}'$, we have  that $\neg H(S, {S}')$ and $\neg H({S}', S)$. Once $H$ is calculated, a frontier of $\A$ can be found in linear time. For example, as $H$ is transitive, we can take one vertex from each ergodic SCC in the graph $\zug{\S(\A),H}$. Note that all frontiers of $\A$ are of the same size, namely the number of ergodic SCCs in this graph. 

\begin{proposition}\label{totally-cover-lemma}
	Consider safe components $S, {S}' \in \S(\A)$ such that $H(S, {S}')$. Then, for every state $p \in S$ there is a state $p'\in {S}'$, such that $p \precsim p'$.
\end{proposition}

Given a frontier $\S$ of $\A$, we define the automaton 
$\B_{\S} = \langle \Sigma, Q_{\S}, q^0_{\S}, \delta_{\S},\alpha_{\S}\rangle$, where $Q_{\S}=\{q\in Q_{\A}: q\in S \text{ for some }S\in {\S} \}$, and the other elements are defined as follows. The initial state $q^0_{\S}$ is chosen such that $q^0_{\S} \sim_\A q^0_\A$. Specifically, if $q^0_\A \in Q_{\S}$, we take $q^0_{\S} = q^0_\A$. Otherwise, by Proposition~\ref{totally-cover-lemma} and the definition of $\S$, there is a state $q' \in Q_{\S}$ such that $q^0_{\A} \precsim q'$, and we take $q^0_{\S}=q'$. The transitions in $\B_{\S}$ are either $\bar{\alpha}$-transitions of $\A$, or $\alpha$-transitions that we add among the safe components in $\S$ in a way that preserves language equivalence. Formally, consider a state $q\in Q_{\S}$ and a letter $\sigma \in \Sigma$. If $\delta^{\bar{\alpha}}_\A(q, \sigma)\neq \emptyset$, then $\delta^{\bar{\alpha}}_\S(q, \sigma) = \delta^{\bar{\alpha}}_\A(q, \sigma)$ and $\delta^{\alpha}_\S(q, \sigma) = \emptyset$.
If $\delta^{\bar{\alpha}}_\A(q,\sigma) = \emptyset$, then
$\delta^{\bar{\alpha}}_\S(q, \sigma) = \emptyset$ and $\delta^\alpha_\S(q, \sigma) = \{ q'\in Q_\S: \mbox{ there is } q''\in \delta^\alpha_\A(q, \sigma)  \mbox{ such that } q' \sim_\A q'' \}$. Note that $\B_{\S}$ is $\alpha$-homogenous.

\begin{example}
{\rm 
	Consider the nice tDCW $\A$ from Figure~\ref{safe example}. By removing the $\alpha$-transitions of $\A$, we get the safe components described in Figure~\ref{A's safe components}. Since $q_2 \precsim q_0$, we have that $\A$ has a single frontier $\S=\{\{q_0,q_1\}\}$. The automaton $\B_{\S}$ appears in Figure \ref{B figure}. As all the states of $\A$ are equivalent, we direct a $\sigma$-labeled $\alpha$-transition to $q_0$ and to $q_1$, for every state with no $\sigma$-labeled transition in $\A$. \hfill \qed}
	
	\vspace{2mm}	
	\begin{minipage}{.4\linewidth}
		\centering
		\includegraphics[width=.7\textwidth]{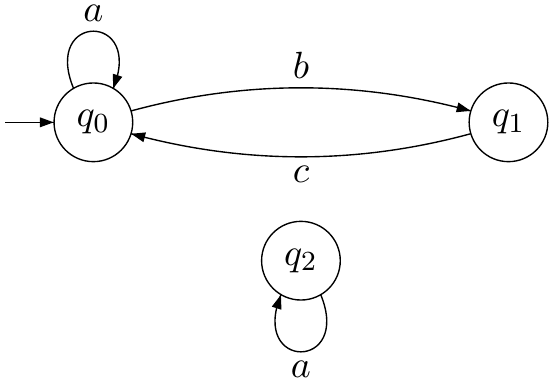}
		\captionof{figure}{The safe components of $\A$.}
		\label{A's safe components}
	\end{minipage}
	\begin{minipage}{.5\linewidth}
		\centering
		\includegraphics[width=.6\textwidth]{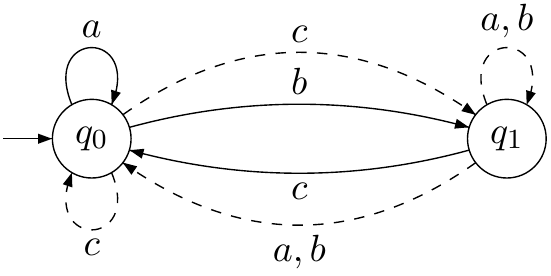}
		\vspace{5mm}
		\captionof{figure}{The tNCW $\B_\S$ for $\S={\{\{q_0,q_1\}\}}$.}
		\label{B figure}	
	\end{minipage}
\end{example}
	\vspace{2mm}

\begin{proposition}\label{A and B are equivalent}
	Let $q$ and $s$ be states of $\A$ and $\B_{\S}$, respectively, with $q \sim_{\A} s$. It holds that $\B_{\S}^s$ is a GFG-tNCW equivalent to $\A^q$.
\end{proposition}

\begin{proposition}
	\label{n sc h}
	For every frontier $\S$, the automaton $\B_{\S}$ is a nice, safe-centralized, and $\alpha$-homogenous GFG-tNCW equivalent to $\A$.
\end{proposition}

\vspace{-2mm}
\paragraph{Step 2: Safe minimization}
Let $\B=\zug{\Sigma,Q,q_0,\delta,\alpha}$ be a nice, safe-centralized, and $\alpha$-homogenous GFG-tNCW.
For 
$q \in Q$, define $[q] = \{ q'\in Q: q \approx_{\B} q' \}$. We define the tNCW $\C = \langle \Sigma, Q_\C, [q_0], \delta_{\C}, \alpha_{\C}\rangle$ as follows. First, $Q_\C=\{[q]: q\in Q\}$. Then, the transition function is such that $\zug{ [q], \sigma, [p]} \in \Delta_\C$ iff there are $q'\in [q]$ and $p'\in [p]$ such that $\langle q', \sigma, p'\rangle \in \Delta$, and $\langle [q], \sigma, [p]\rangle \in \alpha_\C$ iff $\langle q', \sigma, p'\rangle \in \alpha$. Note that $\B$ being $\alpha$-homogenous implies that $\alpha_\C$ is well defined; that is, independent of the choice of $q'$ and $p'$. 
To see why,  assume that $\langle q', \sigma, p'\rangle \in \bar{\alpha}$, and let $q''$ be a state in $[q]$. As $q' \approx_\B q''$, we have, by Proposition \ref{equiv-to-equiv}, that there is $p''\in [p]$ such that $\zug{q'', \sigma, p''}\in \bar{\alpha}$. Thus, as $\B$ is $\alpha$-homogenous,  there is no $\sigma$-labeled $\alpha$-transition from $q''$ in $\B$. In particular, there is no $\sigma$-labeled $\alpha$-transition from $q''$  to a state in $[p]$. Note that, by the above, the tNCW $\C$ is $\alpha$-homogenous. 

\begin{proposition}\label{final equivalent}
	For every $[p]\in Q_{\C}$ and $s\in [p]$, we have that $\C^{[p]}$ is a GFG-tNCW equivalent to~$\B^s$.
\end{proposition}

\begin{proposition}
	\label{n sc sm}
	The GFG-tNCW $\C$ is a nice, safe-centralized, safe-minimal, and $\alpha$-homogenous GFG-tNCW equivalent to $\A$.
\end{proposition}

\begin{example}
{\rm The safe languages of the states $q_0$ and $q_1$ of the GFG-tNCW $\B_\S$ from Figure~\ref{B figure} are different. Thus, $q_0 \not \approx q_1$, and applying safe minimization to $\B_\S$ results in the GFG-tNCW $\C$ identical to $\B_\S$. \hfill \qed}
\end{example}

\section{Canonicity in GFG-NCWs}
\label{sec canonical rep}

In this section we study canonicity for GFG-tNCWs. We first show that the sufficient conditions for minimality of nice GFG-tNCWs specified in 
Theorem~\ref{C is minimal} are necessary.

\begin{theorem}
	\label{minimal automaton has all properties}
	Nice minimal GFG-tNCWs are safe-centralized and safe-minimal.
\end{theorem}

\begin{proof}
	Consider a nice minimal GFG-tNCW $\A$. We argue that if $\A$ is not safe-centralized or not safe-minimal, then it can be minimized further by the minimization construction of \cite{AK19}.
	Assume first that $\A$ is not safe-centralized. Then, there are two different safe components $S, S' \in \S(\A)$ and states $q \in S$ and $q' \in S'$ such that $q \precsim q'$. Then, $H(S, S')$, implying that the safe components in $\S(\A)$ are not a frontier. Then, Step 1 of  the construction minimizes $\A$ further. Indeed, in the transition to the automaton $\B_\S$, at least one safe component in $\S(\A)$ is removed from $\A$ when the frontier $\S$ is computed. 
	Assume now that $\A$ is safe-centralized. Then, for every two different safe components $S, S' \in \S(\A)$, it holds that $\neg H(S, S')$ and $\neg H(S', S)$. Hence, every strict subset of $\S(\A)$ is not a frontier. Thus, $\S(\A)$ is the only frontier of $\A$. Hence,  the automaton $\B_\S$ constructed in Step 1 has $\S=\S(\A)$, and is obtained from $\A$ by adding $\alpha$-transitions that do not change the languages and safe languages of its states. 
	%
	Accordingly, $\B_\S$ is safe-minimal iff $\A$ is safe-minimal. Therefore, if $\A$ is not safe-minimal, then applying Step 2 in the construction to $\B_\S$ merges at least two different states. Hence, also in this case, $\A$ is minimized further.
	\stam{
	 Accordingly, the state spaces of $\A$ and $\B_\S$ coincide, and so do the safe languages of their states, and, by Proposition \ref{A and B are equivalent}, also the languages of their states. Hence, $\B_\S$ is safe-minimal iff $\A$ is safe-minimal. Therefore, if we assume that $\A$ not safe-minimal, then safe-minimizing $\B_\S$ in the transition to the automaton $\C$, as presented in Section~\ref{ak19}, merges at least two different states, thus also in this case, $\A$ is minimized further.
	 }
\end{proof}

We formalize relations between tNCWs by means of {\em isomorphism} and {\em safe isomorphism}.
Consider two tNCWs $\A=\zug{\Sigma,Q_\A,q^0_\A,\delta_\A,\alpha_\A}$ and $\B=\zug{\Sigma,Q_\B,q^0_\B,\delta_\B,\alpha_\B}$, and a bijection $\kappa: Q_\A \to Q_\B$. We say that $\kappa$ is:
\begin{itemize}

\item
{\em $\alpha$-transition respecting}, if $\kappa$ induces a bijection between the $\alpha$-transitions of $\A$ and $\B$. Formally, for all states $q,q' \in Q_\A$ and letter $\sigma \in \Sigma$, we have that  $q' \in \delta^{\alpha_\A}_\A(q, \sigma)$ iff $\kappa(q') \in \delta^{\alpha_\B}_\B(\kappa( q), \sigma)$.


\item
{\em $\bar{\alpha}$-transition respecting}, if $\kappa$ induces a bijection between the $\bar{\alpha}$-transitions of $\A$ and $\B$. Formally, for all states $q,q' \in Q_\A$ and letter $\sigma \in \Sigma$, we have that  $q' \in \delta^{\bar{\alpha}_\A}_\A(q, \sigma)$ iff $\kappa(q') \in \delta^{\bar{\alpha_\B}}_\B(\kappa( q), \sigma)$.

\end{itemize}

Then, $\A$ and $\B$ are \emph{safe isomorphic} if there is a bijection $\kappa: Q_\A \to Q_\B$ that is $\bar{\alpha}$-transition respecting. If, in addition, $\kappa$ is $\alpha$-transition respecting, then $\A$ and $\B$ are \emph{isomorphic}. Note that if $\kappa$ is $\bar{\alpha}$-transition respecting, then for every state $q \in Q_\A$, we have that 
$q$ and $\kappa(q)$ are safe equivalent. Also, if $\kappa$ is both $\alpha$-transition respecting and $\bar{\alpha}$-transition respecting, then for every state $q \in Q_\A$, we have that 
$q \approx \kappa(q)$.

\subsection{Safe isomorphism}

\begin{theorem}
	\label{safe isomorphic}
	Every two equivalent, nice, and minimal GFG-tNCWs are safe isomorphic.
\end{theorem}

\begin{proof}
	Consider two equivalent, nice, and minimal GFG-tNCWs $\A$ and $\B$. By Theorem~\ref{minimal automaton has all properties}, $\A$ is safe-minimal and safe-centralized. Hence, by Lemma~\ref{injection}, there is an injection $\eta: \S(\A) \to \S(\B)$ such that for every safe component $T\in \S(\A)$, it holds that $|T|\leq |\eta(T)|$. For a safe component $T\in \S(\A)$, let $p_T$ be some state in $T$. By Proposition~\ref{there is the same safe}, there are states $q_T \in Q_{\A}$ and $s_T \in Q_{\B}$ such that $p_T \precsim q_T$ and $q_T \approx s_T$. Since $\A$ is safe-centralized, the state $q_T$ is in $T$, and in the proof of Lemma~\ref{injection}, we defined $\eta(T)$ to be the safe component of $s_T$ in $\B$. Likewise, $\B$ is safe-minimal and safe-centralized, and there is an injection $\eta': \S(\B) \to \S(\A)$. 
	The existence of the two injections implies that $|\S(\A)| = |\S(\B)|$. Thus, the injection $\eta$ is actually a bijection. Hence, 
	
	$$|\A| = \sum\limits_{T\in \S(\A)} |T| \leq \sum\limits_{T\in \S(\A)} |\eta(T)| = \sum\limits_{T' \in \S(\B)} |T'| =  |\B|$$
	
	Indeed, the first inequality follows from the fact $|T| \leq |\eta(T)|$, and the second equality follows from the fact that $\eta$ is a bijection.
	Now, as $\A$ and $\B$ are both minimal, we have that $|\A|=|\B|$, and so it follows that for every safe component $T\in \S(\A)$, we have that $|T| = |\eta(T)|$.
	We use the latter fact in order to show that $\eta$ induces a bijection $\kappa: Q_\A \to Q_\B$ that is $\bar{\alpha}$-transition respecting.

	Consider a safe component $T\in \S(\A)$. We define a bijection $\kappa_T: T \to \eta(T)$. The desired bijection $\kappa$ is then the union of the bijections $\kappa_T$ for $T\in \S(\A)$. 
	By Lemma~\ref{injection}, we have that $|T| \leq |\eta(T)|$. The proof of the lemma associates with a safe run  $r_T = q_0, q_1, \ldots q_m$ of $\A$ that traverses all the states in the safe component $T$, a safe run $r_{\eta(T)} = s_0, s_1, \ldots s_m$ of $\B$ that traverses states in $\eta(T)$ and $q_i \approx s_i$, for every $1\leq i\leq m$. Moreover, if $1 \leq i_1,i_2 \leq m$ are such that $q_{i_1} \not \approx q_{i_2}$, then $s_{i_1} \not \approx s_{i_2}$. 
	Now, as $\A$ is safe-minimal, every two states in $T$ are not strongly equivalent. Therefore, the function $\kappa_T$ that maps each state $q_i$ in $r_T$ to the state $s_i$ in $r_{\eta(T)}$ is an injection from $T$ to $\eta(T)$. Thus, as $|T| = |\eta(T)|$, the injection $\kappa_T$ is actually a bijection.

	Clearly, as $\eta: \S(\A) \to \S(\B)$ is a bijection, the function $\kappa$ that is the union of the bijections $\kappa_T$ is a bijection from $Q_\A$ to $Q_\B$. We prove that $\kappa$ is $\bar{\alpha}$-transition respecting.
	Consider states $q,q' \in Q_\A$ and a letter $\sigma \in \Sigma$ such that $\zug{q, \sigma, q'}$ is an $\bar{\alpha}$-transition of $\A$. 
	Let $T$ be $q$'s safe component. By the definition of $\kappa_T$, we have that $q \approx \kappa_T(q)$. By Proposition \ref{equiv-to-equiv}, there is an $\bar{\alpha}$-transition of $\B$ of the form $t = \zug{\kappa(q), \sigma, s'}$, where $q' \approx s'$. As $t$ is an $\bar{\alpha}$-transition of $\B$, we know that $s'$ is in $\eta(T)$. Recall that $\B$ is safe-minimal; in particular, there are no strongly-equivalent states in $\eta(T)$. Hence, $s' = \kappa(q')$, and so $\zug{\kappa(q), \sigma, \kappa(q')}$ is an $\bar{\alpha}$-transition of $\B$. Likewise, if $\zug{\kappa(q), \sigma, \kappa(q')}$ is an $\bar{\alpha}$-transition of $\B$, then $\zug{q, \sigma, q'}$ is an $\bar{\alpha}$-transition of $\A$, and so we are done.
\end{proof}

\subsection{Isomorphism}

Theorem~\ref{safe isomorphic} implies that all nice minimal GFG-tNCWs for a given language are safe isomorphic. We continue and show that it is possible to make these GFG-tNCWs isomorphic. 
We propose two canonical forms that guarantee isomorphism. Both are based on saturating the GFG-NCW with $\alpha$-transitions. One adds as many $\alpha$-transitions as possible, and the second does so in a way that preserves $\alpha$-homogeneity. 

Consider a nice GFG-tNCW $\A=\zug{\Sigma,Q,q_0,\delta,\alpha}$. We say that a triple $\zug{q, \sigma, s} \in Q \times \Sigma \times Q$ is an {\em allowed transition\/} in $\A$ if there is a state $s' \in Q$ such that $s \sim s'$ and $\zug{q, \sigma, s'}\in \Delta$.  Thus, $\zug{q, \sigma, s}$ is allowed if there is a state $s'$ equivalent to $s$ such that $s' \in \delta(q,\sigma)$. We now define two types of $\alpha$-maximality:
\begin{itemize}
\item
We say that $\A$ is {\em $\alpha$-maximal\/} if all allowed transitions in $Q \times \Sigma \times Q$ are in $\Delta$.
\item
We say that $\A$ is {\em $\alpha$-maximal up to homogeneity\/} if $\A$ is $\alpha$-homogenous, and for every state $q \in Q$ and letter $\sigma \in \Sigma$, if $q$ has no outgoing $\sigma$-labeled $\bar{\alpha}$-transitions, then all allowed transitions in $\{q\} \times \{\sigma\} \times Q$ are in $\Delta$. 
\end{itemize} 

Thus, $\alpha$-maximal automata include all allowed transitions, and $\alpha$-maximal up to homogeneity automata include all allowed transitions as long as their inclusion does not conflict with $\alpha$-homogeneity.

\begin{example}
\label{not iso}
	{\rm 
		Recall the minimal GFG-tNCW $\B_{\S}$ appearing in Figure~ \ref{B figure}. 
		The GFG-tNCWs $\C_1$ and $\C_2$ in Figure~\ref{safe iso min} are obtained from $\B_\S$ by removing a $c$-labeled $\alpha$-transition from $q_0$. This does not change the language and result in two minimal equivalent GFG-tNCWs that are safe isomorphic yet
		are not $\alpha$-maximal nor $\alpha$-maximal up to homogeneity. \hfill \qed}
\end{example}
		\begin{figure}[htb]	
			\begin{center}
			\vspace{-4mm}
				\includegraphics[width=.7\textwidth]{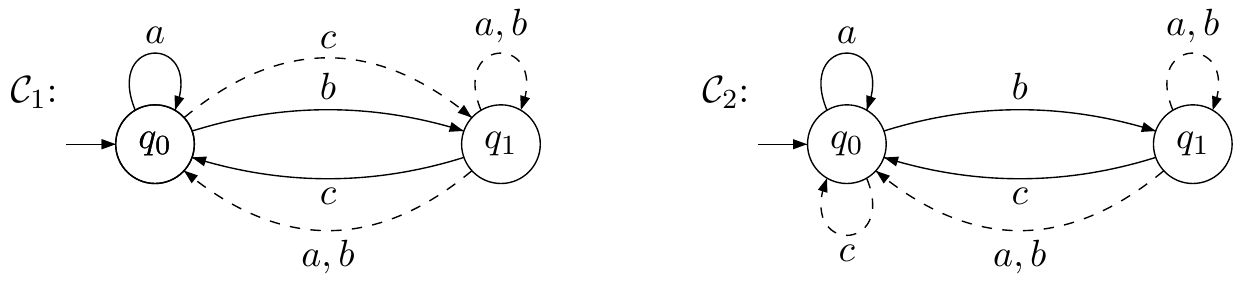}
				\captionof{figure}{Two safe-isomorphic yet not isomorphic minimal equivalent GFG-tNCWs.}
				\label{safe iso min}	
			\end{center}			
		\end{figure}

We now see that both types of $\alpha$-maximality guarantee isomorphism. 

\begin{theorem}
	\label{max imp iso}
	Every two equivalent, nice, minimal, and $\alpha$-maximal GFG-tNCWs are isomorphic.
\end{theorem}

\begin{proof}	
		Consider two equivalent, nice, minimal, and $\alpha$-maximal GFG-tNCWs $\C_1$ and $\C_2$. By Theorem~\ref{safe isomorphic}, we have that $\C_1$ and $\C_2$ are safe isomorphic. Thus, there is a bijection $\kappa: Q_{\C_1} \to Q_{\C_2}$ that is $\bar{\alpha}$-transition respecting. The bijection $\kappa$ was defined such that $q \approx \kappa(q)$, for every state $q\in Q_{\C_1}$. We show that $\kappa$ is also $\alpha$-transition respecting. Let $\zug{q, \sigma, s}$ be an $\alpha$-transition of $\C_1$. Then, as $\kappa$ is $\bar{\alpha}$-transition respecting, and $\zug{q, \sigma, s}$ is not an $\bar{\alpha}$-transition in $\C_1$, the triple $\zug{\kappa(q), \sigma, \kappa(s)}$ cannot be an $\bar{\alpha}$-transition in $\C_2$. We show that $\zug{\kappa(q), \sigma, \kappa(s)}$ is a transition in $\C_2$, and thus it has to be an $\alpha$-transition. As $\C_2$ is nice, in particular total, there is a transition $ \zug{\kappa(q), \sigma, s'}$ in $\C_2$. As $q \sim \kappa(q)$ and both automata are nice, in particular, symantically deterministic, Proposition \ref{pruned-corollary} then implies that $s \sim s'$. Now since $s \sim \kappa(s)$, we get by the transitivity of $\sim$ that $s' \sim \kappa(s)$. Therefore, the existence of the transition $\zug{\kappa(q), \sigma, s'}$ in $\C_2$, implies that the transition $\zug{\kappa(q), \sigma, \kappa(s)} $ is an allowed transition, and so $\alpha$-maximality of $\C_2$ implies that it is also a transition in $\C_2$.  Likewise, if $\zug{\kappa(q), \sigma, \kappa(s)}$ is an $\alpha$-transition in $C_2$, then $\zug{q, \sigma, s}$ is an $\alpha$-transition in $C_1$, and so we are done. 
\end{proof}

\begin{theorem}
	\label{max hom imp iso}
	Every two equivalent, nice, minimal, and $\alpha$-maximal up to homogeneity GFG-tNCWs are isomorphic.
\end{theorem}

\begin{proof}
The proof is identical to that of 	Theorem~\ref{max imp iso}, except that we also have to prove that $\kappa(q)$ has no outgoing $\sigma$-labeled $\bar{\alpha}$-transitions in $\C_2$. To see this, assume by way of contradiction that there is an $\bar{\alpha}$-transition $\zug{\kappa(q), \sigma, s'}$ in $\C_2$. Then, as $q \approx \kappa(q)$, Proposition~\ref{equiv-to-equiv} implies that $q$ has an outgoing $\sigma$-labeled $\bar{\alpha}$-transition in $\C_1$,  contradicting the fact that $\C_1$ is $\alpha$-homogenous. 
\end{proof}

\section{Obtaining Canonical Minimal GFG-tNCWs}
\label{sec obtain}

In this section we show how the two types of canonical minimal GFG-tNCWs can be obtained in polynomial time. 
We start with $\alpha$-maximality up to homogeneity and show that such an $\alpha$-maximization is performed by the minimization construction of \cite{AK19}. We continue with $\alpha$-maximality, show that adding allowed transitions to a GFG-tNCW does not change its language, and conclude that $\alpha$-maximization can be performed on top of the minimization construction of \cite{AK19}.

\subsection{Obtaining canonical minimal $\alpha$-maximal up to homogeneity GFG-tNCWs}

\begin{theorem}
	\label{prop canonical form}
	Consider a nice GFG-tNCW $\A$, and let $\C$ be the minimal GFG-tNCW produced from $\A$ by the minimization construction of \cite{AK19}. Then, $\C$ is $\alpha$-maximal up to homogeneity.
\end{theorem}

\begin{proof}
	Consider the minimization construction of \cite{AK19}. We first show that the safe-centralized GFG-tNCW $\B_{\S}$, defined in Step~1, is $\alpha$-maximal up to homogeneity. Then, we show that $\alpha$-maximality up to homogeneity is maintained in the transition to the GFG-tNCW $\C$, defined in Step~2.
	By Theorem~\ref{n sc h}, we know that $\B_\S$ is $\alpha$-homogenous. Assume that $q$ is a state in $\B_\S$ with no outgoing $\sigma$-labeled $\bar{\alpha}$-transitions, and assume that $\zug{q, \sigma, s}$ is an allowed transition. We need to show that $\zug{q, \sigma, s}$ is a transition in $\B_\S$.
	As $\zug{q, \sigma, s}$ is an allowed transition, there is a transition $\zug{q, \sigma, s'}$ in $\B_\S$ with $s \sim_{\B_\S} s'$, and by the assumption, $\zug{q, \sigma, s'}$ has to be an $\alpha$-transition.
	By the definition of the transition function of $\B_\S$, we have that $s' \sim_\A q'$ for some state $q' \in \delta^{\alpha}_\A(q, \sigma)$. As $\A$ is semantically deterministic, we get that the state $s'$ is $\A$-equivalent to every state in $\delta^{\alpha}_\A(q, \sigma)$. So again, by the definition of the transition function of $\B_\S$, we can write $\delta^{\alpha}_\S(q, \sigma) = \{p \in Q_\S: p\sim_\A q'\}$. 
	Now, as $s \sim_{\B_\S} s'$, Proposition \ref{A and B are equivalent} implies that $L(\A^{s}) = L(\B^{s}_\S) = L(\B^{s'}_\S) = L(\A^{s'})$; that is, $s\sim_\A s'$, and since $s'\sim_\A q'$, we get by the transitivity of $\sim_\A$ that $s \sim_\A q'$, and so $\zug{q, \sigma, s}$ is a transition in $\B_\S$.
	
	We show next that the GFG-tNCW $\C$ is $\alpha$-maximal up to homogeneity. By Theorem~\ref{n sc sm}, we have that $\C$ is $\alpha$-homogenous. 
	Assume that $[q]$ is a state in $\C$ with no outgoing $\sigma$-labeled $\bar{\alpha}$-transitions, and assume that $\zug{[q], \sigma, [s]}$ is an allowed transition. We need to show that $\zug{[q], \sigma, [s]}$ is a transition in $\C$.
	As $\zug{[q], \sigma, [s]}$ is an allowed transition, there is a transition $\zug{[q], \sigma, [s']}$ in $\C$ with $[s] \sim [s']$.
	Thus, by Proposition \ref{final equivalent}, we have that $L(\B^{s}_\S) = L(\C^{[s]}) = L(\C^{[s']}) = L(\B^{s'}_\S)$; that is, $s'\sim_{\B_\S} s$.  
	By the assumption, $\zug{[q], \sigma, [s']}$ has to be an $\alpha$-transition. Therefore, by the definition of $\C$, there are states $q''\in [q]$ and $s''\in [s']$, such that $\langle q'', \sigma, s''\rangle$ is an $\alpha$-transition in $\B_\S$. Now, by transitivity of $\sim_{\B_\S}$ and the fact that $s'' \sim_{\B_\S} s'$, we get that $s''\sim_{\B_\S} s$. Finally, as $\B_\S$ is $\alpha$-homogenous, we get that $q''$ has no outgoing $\sigma$-labeled $\bar{\alpha}$-transitions in $\B_\S$, and so by the $\alpha$-maximality up to homogeneity of $\B_\S$, we have that $\langle q'', \sigma, s\rangle$ is a transition in $\B_\S$. Therefore, by the definition of $\C$,  we have that $\langle  [q], \sigma, [s]\rangle$ is a transition in $\C$, and we are done.
\end{proof}

We can thus conclude with the following.

\begin{theorem}
	Every GFG-tNCW $\A$ can be canonized into a nice minimal $\alpha$-maximal up to homogeneity GFG-tNCW in polynomial time.
\end{theorem}

\subsection{Obtaining canonical minimal $\alpha$-maximal GFG-tNCWs}

Consider a nice GFG-tNCW $\A = \zug{\Sigma, Q, q_0, \delta, \alpha}$. We say that a set of triples  $\E \subseteq Q\times \Sigma \times Q$ is an \emph{allowed set} if all the triples in it are allowed transitions in $\A$.  
For every set $\E \subseteq Q\times \Sigma \times Q$, we define the tNCW $\A_\E = \zug{\Sigma, Q_, q_0, \delta_\E, \alpha_\E}$, where $\Delta_\E = \Delta \cup \E$ and $\alpha_\E = \alpha \cup \E$. Clearly, as $\A$ and $\A_\E$ have the same set of states and the same set of $\bar{\alpha}$-transitions, they are safe equivalent. 


In Propositions~\ref{C and CE are equivalent} and~\ref{ae is nice} below, we prove that for every allowed set $\E$, we have that $\A_{\E}$ is a nice GFG-tNCW equivalent to $\A$. 
We first extend Proposition~\ref{pruned-corollary} to the setting of $\A$ and $\A_\E$:

\begin{proposition}\label{pruned-corollaryC}
	Consider states $q$ and $s$ of $\A$ and $\A_\E$, respectively, a letter $\sigma \in \Sigma$, and transitions 
	$\zug{q, \sigma, q'}$ and $\langle s, \sigma, s'\rangle$ of $\A$ and $\A_\E$, respectively. 
	If $q \sim_{\A} s$, then $q' \sim_{\A} s'$.
\end{proposition}

\begin{proof}
	If $\langle s, \sigma, s'\rangle \notin \E$, then, by the definition of $\Delta_{\E}$, it is also a transition of $\A$. Hence, since $q \sim_{\A} s$ and $\A$ is nice, in particular, semantically deterministic, Proposition~\ref{pruned-corollary} implies  that $q' \sim_{\A} s'$. 
	If $\langle s, \sigma, s'\rangle \in \E$, then, by the definition of $\Delta_\E$, it is an allowed transition of $\A$.  Therefore, there is a state $p'\in Q$ such that $s' \sim_\A p'$ and $\zug{s, \sigma, p'}\in \Delta$. As $q\sim_\A s$ and $\A$ is semantically deterministic, Proposition \ref{pruned-corollary}  implies that $q'\sim_{\A} p'$. Therefore, using the fact that $p'\sim_\A s'$, the transitivity of $\sim_\A$ implies that $q' \sim_\A s'$, and so we are done.
\end{proof}

\begin{proposition}\label{C and CE are equivalent}
	Let $p$ and $s$ be states of $\A$ and $\A_\E$, respectively, with $p \sim_{\A} s$. Then, $\A^s_\E$ is a GFG-tNCW equivalent to $\A^p$.
\end{proposition}

\begin{proof}
	We first prove that $L(\A^s_\E) \subseteq L(\A^p)$. Consider a word $w=\sigma_1\sigma_2\ldots \in L(\A^s_\E)$, and let $s_0,s_1,s_2,\ldots $ be an accepting run of $\A^s_\E$ on $w$. Then, there is $i\geq 0$ such that $s_i,s_{i+1},\ldots $ is a safe run of $\A^{s_i}_\E$ on the suffix $w[i+1, \infty]$. Let $p_0,p_1,\ldots p_i$ be a run of $\A^p$ on the prefix $w[1, i]$. Since $p_0 \sim_{\A} s_0$, we get, by an iterative application of Proposition~\ref{pruned-corollaryC}, that $p_i \sim_{\A} s_i$. In addition, as the run of $\A^{s_i}_\E$ on the suffix $w[i+1, \infty]$ is safe, it is also a safe run of $\A^{s_i}$. Hence, $w[i+1, \infty] \in L(\A^{p_i})$, and thus $p_0,p_1,\ldots, p_i$ can be extended to an accepting run of $\A^p$ on $w$. 
	
	Next, as $\A$ is nice, all of its states are GFG, in particular, there is a strategy $f^s$ witnessing $\A^s$'s GFGness. Recall that $\A$ is embodied in $\A_\E$. Therefore, every run in $\A$ exists also in $\A_\E$. Thus, as $p \sim_\A s$, we get that for every word $w\in L(\A^p)$, the run $f^s(w)$ is an accepting run of $\A^s$ on $w$, and thus is also an accepting run of $\A^s_\E$ on $w$. Hence, $L(\A^p) \subseteq L(\A^s_\E)$  and $f^s$ witnesses $\A^s_\E$'s GFGness.	
	
\end{proof}

\begin{proposition}
\label{ae is nice}
	For every allowed set $\E$, the GFG-tNCW $\A_\E$ is nice.
\end{proposition}

\begin{proof}
		It is easy to see that the fact $\A$ is nice implies that $\A_\E$ is normal and safe deterministic. Also, as $\A$ is embodied in $\A_\E$ and both automata have the same state-space and initial states, then all the states in $\A_\E$ are reachable.
	Finally, Proposition~\ref{C and CE are equivalent} implies that all the states in $\A_\E$ are GFG. To conclude that $\A_\E$ is nice, we prove below that it is semantically deterministic.  Consider transitions $\langle q, \sigma, s_1\rangle$ and $\langle q, \sigma, s_2\rangle$ in $\Delta_{\E}$. We need to show that $s_1 \sim_{\A_\E} s_2$. By the definition of $\Delta_\E$, 	there are transitions $\langle q, \sigma, s'_1\rangle$ and $\langle q, \sigma, s'_2\rangle$ in $\Delta$ for states $ s'_1$ and $s'_2$ such that $s_1 \sim_\A s'_1$ and $s_2 \sim_\A s'_2$. As $\A$ is nice, in particular, semantically deterministic, we have that $s'_1 \sim_\A s'_2$. Hence, as $s_1 \sim_\A s'_1$ and $s'_2 \sim_\A s_2$, we get by the transitivity of $\sim_\A$ that $s_1 \sim_\A s_2$. Then, Proposition~\ref{C and CE are equivalent} implies that $L(\A^{s_1}) = L(\A^{s_1}_{\E})$ and $L(\A^{s_2}) = L(\A^{s_2}_{\E})$, and so we get that $s_1 \sim_{\A_{\E}} s_2$. Thus, $\A_{\E}$ is semantically deterministic. 
\end{proof}

Let $\C$ be a nice minimal GFG-tNCW equivalent to $\A$, and let $\hat{\E}$ be the set of all allowed transitions in $\C$. By Propositions~\ref{C and CE are equivalent} and~\ref{ae is nice}, we have that $\C_{\hat{\E}}$ is a nice minimal GFG-tNCW equivalent to $\A$. Below we argue that it is also $\alpha$-maximal. 

\begin{proposition}
	Let $\C$ be a nice GFG-tNCW, and let $\hat{\E}$ be the set of all allowed transitions in $\C$. Then, $\C_{\hat{\E}}$  is $\alpha$-maximal.
\end{proposition}

\begin{proof}
	Let $\C=\langle \Sigma, Q, q_0, \delta, \alpha \rangle$, and consider an allowed transition $\zug{q, \sigma, s} \in Q\times \Sigma \times Q$ in $\C_{\hat{\E}}$. We prove that $\zug{q, \sigma, s}$ is an allowed transition also in $\C$. Hence, it is in ${\hat{\E}}$, and thus is a transition in $\C_{\hat{\E}}$. 
	
	By the definition of allowed transitions, there is a state $s'\in Q$ with $s \sim_{\C_{\hat{\E}}} s'$ such that $s' \in \delta_{\hat{\E}}(q, \sigma)$. Proposition~\ref{C and CE are equivalent} implies that $L(\C^s) = L(\C^s_{\hat{\E}}) = L(\C^{s'}_{\hat{\E}}) =  L(\C^{s'})$, and thus $s \sim_\C s'$. Also, by the definition of $\delta_{\hat{\E}}$, there is a state $s''\in Q$ such that $s''\sim_ \C s'$ and $s'' \in \delta(q, \sigma)$. Therefore, as the transitivity of $\sim_\C$ implies that $s\sim_\C s''$, we have that $\zug{q, \sigma, s}$ is also an allowed transition in $\C$, and we are done.
	\end{proof}

Since the relation $\sim$ can be calculated in polynomial time \cite{HKR02,KS15}, and so checking if a triple in $Q \times \Sigma \times Q$ is an allowed transition can be done in polynomial time, then applying $\alpha$-maximization on top of the minimization construction of \cite{AK19} is still polynomial. We can thus conclude with the following.

\begin{theorem}
	Every GFG-tNCW $\A$ can be canonized into a nice minimal $\alpha$-maximal GFG-tNCW in polynomial time.
\end{theorem}

\begin{example}
{\rm By applying $\alpha$-maximization to the GFG-tNCW $\B_\S$ , we obtained the $\alpha$-maximal GFG-tNCW $\C_{\hat{\E}}$ appearing in Figure~\ref{alpha maximal fig} \hfill \qed}
\end{example}

		\begin{figure}[htb]	
			\begin{center}
				\vspace{-7mm}
				\includegraphics[width=.32\textwidth]{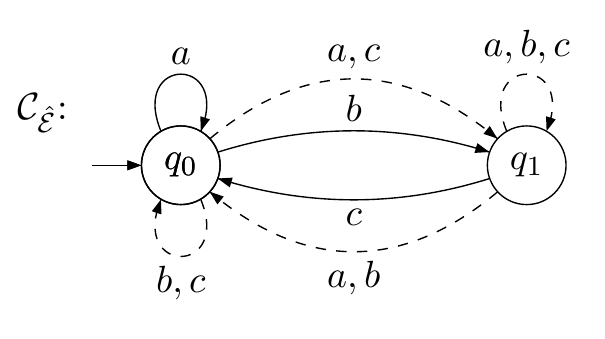}
				\captionof{figure}{The $\alpha$-maximal GFG-tNCW for the GFG-tNCW $\B_\S$ in Figure~\ref{B figure}.}
				\label{alpha maximal fig}	
			\end{center}	
		\end{figure}

\section{Canonicity in tDCW and tDBW}
\label{sc canonicity determinisic}

For deterministic automata with state-based acceptance, an analogue definition of isomorphism between automata $\A$ and $\B$ with acceptance conditions $\alpha_\A \subseteq Q_\A$ and $\alpha_\B \subseteq Q_\B$, seeks a bijection $\kappa: Q_\A \to Q_\B$ such that for every $q\in Q_\A$, we have that $q \in \alpha_\A$ iff $\kappa(q) \in \alpha_\B$, and for  every letter $\sigma \in \Sigma$, and state $q' \in Q_\A$, we have that $q' \in \delta_\A(q,\sigma)$ iff $\kappa(q') \in \delta_\B(\kappa(q),\sigma)$. It is easy to see that the DCWs $\A_1$ and $\A_2$ from Figure~\ref{2min dcws} are not isomorphic, which is a well known property of DCWs and DBWs \cite{Kup15}. In Theorem~\ref{no can gfg} below, we extend the ``no canonicity" result to GFG-NCWs.

\begin{theorem}
\label{no can gfg}
Nice, equivalent, and minimal GFG-NCWs need not be isomorphic.
\end{theorem}

\begin{proof}
Consider the language $L=(a+b)^* \cdot (a^\omega + b^\omega)$. In Figure~\ref{2min dcws}, we described  
the non-isomorphic DCWs $\A_1$ and $\A_2$ for $L$. The DCWs $\A$ and $\B$ can be viewed as nice GFG-NCWs.
It is not hard to see that there is no $2$-state GFG-NCW for $L$, implying that $\A$ and $\B$ are nice, equivalent, and minimal GFG-NCWs that are not isomorphic, as required. 
\end{proof}

In Example~\ref{not iso} we saw that nice, equivalent, and minimal GFG-tNCWs need not be isomorphic too, yet they may be made isomorphic by $\alpha$-maximization. For the GFG-NCWs in the proof of Theorem~\ref{no can gfg}, this does not work for every definition of $\alpha$-maxization that makes sense: we cannot add transitions and make the automata isomorphic. This suggests that the consideration of automata with transition-based acceptance is more crucial for canonization than the consideration of GFG automata, and makes the study of canonization for tDCWs very interesting. In particular, unlike the case of GFG automata, here results on tDCWs immediately apply also to tDBWs.
We start with some bad news, showing that there is no canonicity also in the transition-based setting. 

\begin{theorem}
\label{tD not isomorphic}
Nice, equivalent, and minimal tDCWs and tDBWs need not be isomorphic. 
\end{theorem}

\begin{proof}
The GFG-tNCW $\B_\S$ from Figure~\ref{B figure} is DBP. In Figure~\ref{tdcw not iso min} below, we describe two tDCWs obtained from it by two different prunnings. It is not hard to see that both tDCWs are equivalent to $\B_\S$, yet are not isomorphic.
		\begin{figure}[htb]	
			\begin{center}
			\hspace{-4mm}
				\includegraphics[width=.7\textwidth]{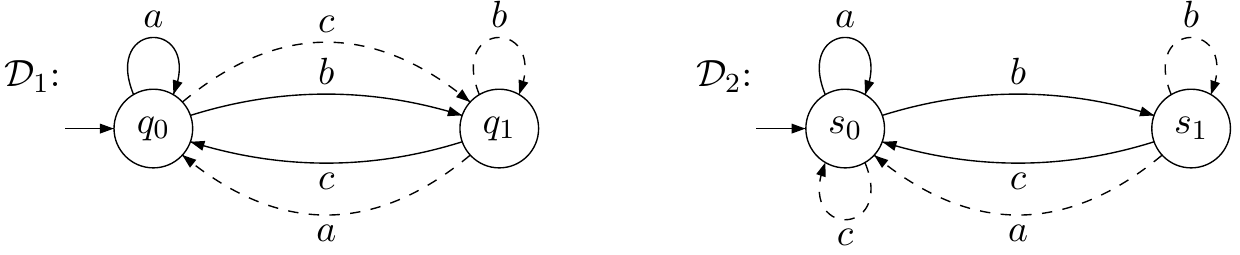}
				\captionof{figure}{Two non-isomorphic equivalent minimal nice tDCWs, obtained by different prunnings of $\B_\S$.}
				\label{tdcw not iso min}	
			\end{center}			
		\end{figure}
		
By removing the $a$-labeled transitions from the tDCWs in Figure~\ref{tdcw not iso min}, we obtain a simpler example.  Consider the tDCWs $\D'_1$ and $\D'_2$ in Figure~\ref{tdcw not iso min no a}. It is easy to see that $L(\D'_1)=L(\D'_2)=(b+c)^* \cdot (b \cdot c)^\omega$. Clearly, there is no single-state tDCW for this language. Also, the tDCWs are not isomorphic, as a candidate bijection $\kappa$ has to be $\bar{\alpha}$-transition respecting, and thus have $\kappa(q_0)=s_0$ and $\kappa(q_1)=s_1$, yet then it is not $\alpha$-transition respecting. By dualizing the acceptance condition of $\D'_1$ and $\D'_2$, we obtain two non-isomorphic tDBWs for the complement language, of all words with infinitely many occurrences of $bb$ or $cc$.
\end{proof}

		\begin{figure}[htb]	
			\begin{center}
			\hspace{-6mm}
				\includegraphics[width=.7\textwidth]{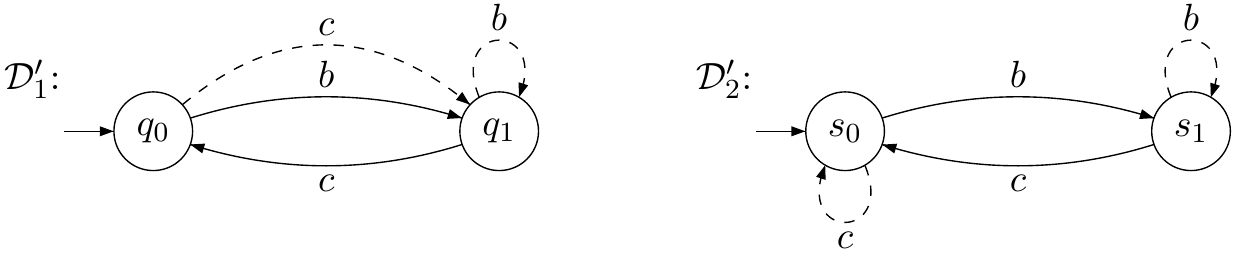}
				\captionof{figure}{Two nice, minimal, equivalent, and non-isomorphic tDCWs}
				\label{tdcw not iso min no a}	
			\end{center}			
		\end{figure}

The GFG-NCWs used in the proof of Theorem~\ref{no can gfg} cannot be made isomorphic by changing membership of states in $\alpha$ or by adding transitions. Likewise,  since tDCWs cannot be $\alpha$-maximized, as adding transitions conflicts with determinism, the tDCWs used in the proof of Theorem~\ref{tD not isomorphic} cannot be made isomorphic either. Hence, we have the following.

\begin{theorem}
There is no canonicity for minimal GFG-NCWs and for minimal tDCWs.
\end{theorem}

The tDCWs in the proof of Theorem~\ref{tD not isomorphic} are safe isomorphic, 
We continue and study safe-isomorphism between minimal tDCWs. Here too, we restrict attention to {\em nice\/} minimal tDCWs. Note that here, some of the properties of nice GFG-tNCWs are trivial: being minimal and deterministic, then clearly all states are reachable and GFG, the automata are semantically deterministic and safe deterministic, and we only have to make them normal by classifying transitions between safe components as $\alpha$-transitions.

\stam{
\begin{theorem}
\label{not safe iso}
Consider an $\omega$-regular language $L$. If $L$ has two nice minimal tDCWs that are not safe isomorphic, then every minimal GFG-tNCW for $L$ is not DBP.
\end{theorem}

\begin{proof}
Consider a language $L$ that has two nice minimal tDCWs $\A_1$ and $\A_2$ that are not safe isomorphic. Let $\A$ be a minimal GFG-tNCW for $L$. We distinguish between two cases:
\begin{enumerate}
\item
If $|\A|<|\A_1|$, then $\A$ is not DBP, as otherwise by pruning $\A$ we would have obtained a tDCW for $L$ with less than $|\A_1|$ states, contradicting the minimality of $\A_1$. 
\item
If $|\A|=|\A_1|$, then $\A_1$ and $\A_2$ are nice minimal GFG-tNCWs for $L$, and hence by Theorem~\ref{safe isomorphic}, they are safe isomorphic. 
\end{enumerate}
\end{proof}
}

We say that an $\omega$-regular language $L$ is {\em tDCW-positive\/} if a minimal tDCW for $L$ is not bigger than a minimal GFG-tNCW for $L$. Thus, tDCWs for $L$ are as succinct as GFG-tNCWs for it. 
\begin{theorem}
\label{not safe iso}
Consider an $\omega$-regular language $L$. If $L$ is tDCW-positive, then 
	every two nice and minimal tDCWs for $L$ are safe isomorphic.
\end{theorem}

\begin{proof}
Consider a language $L$ that is tDCW-positive, and consider two nice minimal tDCWs $\A_1$ and $\A_2$ for $L$. 
Since $L$ is tDCW-positive, then $\A_1$ and $\A_2$ are also nice minimal GFG-tNCWs for $L$. Hence, by Theorem~\ref{safe isomorphic}, they are safe isomorphic. 
\end{proof}

Note that safe isomorphism for $\omega$-regular languages that are not tDCW-positive is left open. Theorem~\ref{not safe iso} suggests that searching for a language $L$ that has two minimal tDCWs that are not safe isomorphic, we can restrict attention to languages that are not tDCW-positive. Such languages are not natural. Moreover, their canonicity is less crucial, as working with a minimal GFG-tNCW for them is more appealing.  Examples of languages that are not tDCW-positive can be found in \cite{KS15}, where it was shown that GFG-tNCWs may be exponentially more succinct than tDCWs.

\section{Glossary}
\label{app glos}

All notations and definitions refer to a  GFG-tNCW $\A = \langle \Sigma, Q, q_0, \delta, \alpha  \rangle$.

\paragraph{Relations between states}
\begin{itemize}
	\item
	Two states $q,s\in Q$ are \emph{equivalent}, denoted $q \sim s$, if $L(\A^q) = L(\A^s)$. 
	\item
	Two states $q,s\in Q$ are \emph{safe equivalent}  if $L_{safe}(\A^q) = L_{safe}(\A^s)$. 
	\item
	Two states $q,s\in Q$ are \emph{strongly-equivalent}, denoted $q \approx  s$, if $q \sim s$ and $L_{safe}(\A^q) = L_{safe}(\A^s)$. 
	\item
	A state $q \in Q$ is \emph{subsafe-equivalent to} a state $s$, denoted $q\precsim s$, if $q \sim s$ and $L_{safe}(\A^q) \subseteq L_{safe}(\A^s)$. 
\end{itemize}

\paragraph{Properties of a GFG-tNCW}
\begin{itemize}
	\item
	$\A$ is \emph{semantically deterministic} if for every state $q\in Q$ and letter $\sigma \in \Sigma$, all the $\sigma$-successors of $q$ are equivalent: for every two states $s, s'\in \delta(q,\sigma)$, we have that $s \sim s'$.
	\item
	$\A$ is \emph{safe deterministic} if by removing its $\alpha$-transitions, we get a (possibly not total)  deterministic automaton. Thus, for every state $q\in Q$ and letter $\sigma\in \Sigma$, it holds that $|\delta^{\bar{\alpha}}(q, \sigma)|\leq 1$.
	\item
	$\A$ is \emph{normal} if there are no $\bar{\alpha}$-transitions connecting different safe components. That is,
	for all states $q$ and $s$ of $\A$, if there is a path of $\bar{\alpha}$-transitions from $q$ to $s$, then there is also a path of $\bar{\alpha}$-transitions from $s$ to $q$. 
	\item
	$\A$ is {\em nice\/} if all the states in $\A$ are reachable and GFG, and $\A$ is normal, safe deterministic, and semantically deterministic. 
	\item
	$\A$ is \emph{$\alpha$-homogenous} if for every state $q\in Q$ and letter $\sigma \in \Sigma$, either $\delta^\alpha(q, \sigma) =\emptyset$ or $\delta^{\bar{\alpha}}(q, \sigma) = \emptyset$. 
	\item
	$\A$ is \emph{safe-minimal} if it has no strongly-equivalent states. 
	\item
	$\A$ is \emph{safe-centralized} if for every two states $q, s\in Q$, if $q \precsim s$, then $q$ and $s$ are in the same safe component of $\A$. 
\end{itemize}

\newpage

\bibliographystyle{eptcs}
\bibliography{ba}

\begin{thebibliography}{10}
\providecommand{\bibitemdeclare}[2]{}
\providecommand{\surnamestart}{}
\providecommand{\surnameend}{}
\providecommand{\urlprefix}{Available at }
\providecommand{\url}[1]{\texttt{#1}}
\providecommand{\href}[2]{\texttt{#2}}
\providecommand{\urlalt}[2]{\href{#1}{#2}}
\providecommand{\doi}[1]{doi:\urlalt{http://dx.doi.org/#1}{#1}}
\providecommand{\bibinfo}[2]{#2}

\bibitemdeclare{inproceedings}{AK19}
\bibitem{AK19}
\bibinfo{author}{B.~\surnamestart {Abu Radi}\surnameend} \&
  \bibinfo{author}{O.~\surnamestart Kupferman\surnameend}
  (\bibinfo{year}{2019}): \emph{\bibinfo{title}{Minimizing {GFG}
  Transition-Based Automata}}.
\newblock In: {\sl \bibinfo{booktitle}{Proc.\ 46th Int. Colloq. on Automata,
  Languages, and Programming}}, {\sl \bibinfo{series}{LIPIcs}}
  \bibinfo{volume}{132}, \bibinfo{publisher}{Schloss Dagstuhl - Leibniz-Zentrum
  fuer Informatik}, pp. \bibinfo{pages}{100:1--100:16},
  \doi{10.4230/LIPIcs.ICALP.2019.100}.

\bibitemdeclare{inproceedings}{BK18}
\bibitem{BK18}
\bibinfo{author}{M.~\surnamestart Bagnol\surnameend} \&
  \bibinfo{author}{D.~\surnamestart Kuperberg\surnameend}
  (\bibinfo{year}{2018}): \emph{\bibinfo{title}{B{\"{u}}chi Good-for-Games
  Automata Are Efficiently Recognizable}}.
\newblock In: {\sl \bibinfo{booktitle}{Proc. 38th Conf. on Foundations of
  Software Technology and Theoretical Computer Science}}, {\sl
  \bibinfo{series}{LIPIcs}} \bibinfo{volume}{122}, \bibinfo{publisher}{Schloss
  Dagstuhl - Leibniz-Zentrum fuer Informatik}, pp.
  \bibinfo{pages}{16:1--16:14}, \doi{10.4230/LIPIcs.FSTTCS.2018.16}.

\bibitemdeclare{inproceedings}{BKKS13}
\bibitem{BKKS13}
\bibinfo{author}{U.~\surnamestart Boker\surnameend},
  \bibinfo{author}{D.~\surnamestart Kuperberg\surnameend},
  \bibinfo{author}{O.~\surnamestart Kupferman\surnameend} \&
  \bibinfo{author}{M.~\surnamestart Skrzypczak\surnameend}
  (\bibinfo{year}{2013}): \emph{\bibinfo{title}{Nondeterminism in the Presence
  of a Diverse or Unknown Future}}.
\newblock In: {\sl \bibinfo{booktitle}{{ICALP} {(2)}}}, {\sl
  \bibinfo{series}{Lecture Notes in Computer Science}} \bibinfo{volume}{7966},
  \bibinfo{publisher}{Springer}, pp. \bibinfo{pages}{89--100},
  \doi{10.1007/978-3-642-39212-2_11}.

\bibitemdeclare{inproceedings}{BKS17}
\bibitem{BKS17}
\bibinfo{author}{U.~\surnamestart Boker\surnameend},
  \bibinfo{author}{O.~\surnamestart Kupferman\surnameend} \&
  \bibinfo{author}{M.~\surnamestart Skrzypczak\surnameend}
  (\bibinfo{year}{2017}): \emph{\bibinfo{title}{How Deterministic are
  {G}ood-{F}or-{G}ames Automata?}}
\newblock In: {\sl \bibinfo{booktitle}{Proc. 37th Conf. on Foundations of
  Software Technology and Theoretical Computer Science}}, {\sl
  \bibinfo{series}{Leibniz International Proceedings in Informatics
  (LIPIcs)}}~\bibinfo{volume}{93}, pp. \bibinfo{pages}{18:1--18:14},
  \doi{10.4230/LIPIcs.FSTTCS.2017.18}.

\bibitemdeclare{inproceedings}{Buc62}
\bibitem{Buc62}
\bibinfo{author}{J.R. \surnamestart B{\"u}chi\surnameend}
  (\bibinfo{year}{1962}): \emph{\bibinfo{title}{On a Decision Method in
  Restricted Second Order Arithmetic}}.
\newblock In: {\sl \bibinfo{booktitle}{Proc. Int. Congress on Logic, Method,
  and Philosophy of Science. 1960}}, \bibinfo{publisher}{Stanford University
  Press}, pp. \bibinfo{pages}{1--12}.

\bibitemdeclare{conference}{Col09}
\bibitem{Col09}
\bibinfo{author}{Th. \surnamestart Colcombet\surnameend}
  (\bibinfo{year}{2009}): \emph{\bibinfo{title}{The theory of stabilisation
  monoids and regular cost functions}}.
\newblock In: {\sl \bibinfo{booktitle}{Proc.\ 36th Int. Colloq. on Automata,
  Languages, and Programming}}, {\sl \bibinfo{series}{Lecture Notes in Computer
  Science}} \bibinfo{volume}{5556}, \bibinfo{publisher}{Springer}, pp.
  \bibinfo{pages}{139--150}, \doi{10.1007/978-3-642-02930-1_12}.

\bibitemdeclare{inproceedings}{DLFMRX16}
\bibitem{DLFMRX16}
\bibinfo{author}{A.~\surnamestart Duret-Lutz\surnameend},
  \bibinfo{author}{A.~\surnamestart Lewkowicz\surnameend},
  \bibinfo{author}{A.~\surnamestart Fauchille\surnameend}, \bibinfo{author}{Th.
  \surnamestart Michaud\surnameend}, \bibinfo{author}{E.~\surnamestart
  Renault\surnameend} \& \bibinfo{author}{L.~\surnamestart Xu\surnameend}
  (\bibinfo{year}{2016}): \emph{\bibinfo{title}{Spot 2.0 --- a framework for
  {LTL} and $\omega$-automata manipulation}}.
\newblock In: {\sl \bibinfo{booktitle}{14th Int. Symp. on Automated Technology
  for Verification and Analysis}}, {\sl \bibinfo{series}{Lecture Notes in
  Computer Science}} \bibinfo{volume}{9938}, \bibinfo{publisher}{Springer}, pp.
  \bibinfo{pages}{122--129}, \doi{10.1007/978-3-319-46520-3_8}.

\bibitemdeclare{inproceedings}{GL02}
\bibitem{GL02}
\bibinfo{author}{D.~\surnamestart Giannakopoulou\surnameend} \&
  \bibinfo{author}{F.~\surnamestart Lerda\surnameend} (\bibinfo{year}{2002}):
  \emph{\bibinfo{title}{From States to Transitions: Improving Translation of
  {LTL} Formulae to {B}{\"u}chi Automata}}.
\newblock In: {\sl \bibinfo{booktitle}{Proc. 22nd International Conference on
  Formal Techniques for Networked and Distributed Systems}}, {\sl
  \bibinfo{series}{Lecture Notes in Computer Science}} \bibinfo{volume}{2529},
  \bibinfo{publisher}{Springer}, pp. \bibinfo{pages}{308--326},
  \doi{10.1007/3-540-36135-9_20}.

\bibitemdeclare{article}{HKR02}
\bibitem{HKR02}
\bibinfo{author}{T.A. \surnamestart Henzinger\surnameend},
  \bibinfo{author}{O.~\surnamestart Kupferman\surnameend} \&
  \bibinfo{author}{S.~\surnamestart Rajamani\surnameend}
  (\bibinfo{year}{2002}): \emph{\bibinfo{title}{Fair simulation}}.
\newblock {\sl \bibinfo{journal}{Information and Computation}}
  \bibinfo{volume}{173}(\bibinfo{number}{1}), pp. \bibinfo{pages}{64--81},
  \doi{10.1006/inco.2001.3085}.

\bibitemdeclare{inproceedings}{HP06}
\bibitem{HP06}
\bibinfo{author}{T.A. \surnamestart Henzinger\surnameend} \&
  \bibinfo{author}{N.~\surnamestart Piterman\surnameend}
  (\bibinfo{year}{2006}): \emph{\bibinfo{title}{Solving Games without
  Determinization}}.
\newblock In: {\sl \bibinfo{booktitle}{Proc. 15th Annual Conf. of the European
  Association for Computer Science Logic}}, {\sl \bibinfo{series}{Lecture Notes
  in Computer Science}} \bibinfo{volume}{4207}, \bibinfo{publisher}{Springer},
  pp. \bibinfo{pages}{394--410}, \doi{10.1007/11874683_26}.

\bibitemdeclare{incollection}{Hop71}
\bibitem{Hop71}
\bibinfo{author}{J.E. \surnamestart Hopcroft\surnameend}
  (\bibinfo{year}{1971}): \emph{\bibinfo{title}{An $n \log n$ algorithm for
  minimizing the states in a finite automaton}}.
\newblock In \bibinfo{editor}{Z.~\surnamestart Kohavi\surnameend}, editor: {\sl
  \bibinfo{booktitle}{The Theory of Machines and Computations}},
  \bibinfo{publisher}{Academic Press}, pp. \bibinfo{pages}{189--196},
  \doi{10.1016/B978-0-12-417750-5.50022-1}.

\bibitemdeclare{inproceedings}{KS15}
\bibitem{KS15}
\bibinfo{author}{D.~\surnamestart Kuperberg\surnameend} \&
  \bibinfo{author}{M.~\surnamestart Skrzypczak\surnameend}
  (\bibinfo{year}{2015}): \emph{\bibinfo{title}{On Determinisation of
  Good-for-Games Automata}}.
\newblock In: {\sl \bibinfo{booktitle}{Proc.\ 42nd Int. Colloq. on Automata,
  Languages, and Programming}}, pp. \bibinfo{pages}{299--310},
  \doi{10.1007/978-3-662-47666-6_24}.

\bibitemdeclare{article}{Kup15}
\bibitem{Kup15}
\bibinfo{author}{O.~\surnamestart Kupferman\surnameend} (\bibinfo{year}{2015}):
  \emph{\bibinfo{title}{Automata Theory and Model Checking}}.
\newblock {\sl \bibinfo{journal}{Handbook of {T}heoretical {C}omputer
  {S}cience}}.

\bibitemdeclare{article}{KSV06}
\bibitem{KSV06}
\bibinfo{author}{O.~\surnamestart Kupferman\surnameend},
  \bibinfo{author}{S.~\surnamestart Safra\surnameend} \& \bibinfo{author}{M.Y.
  \surnamestart Vardi\surnameend} (\bibinfo{year}{2006}):
  \emph{\bibinfo{title}{Relating word and tree automata}}.
\newblock {\sl \bibinfo{journal}{Ann. Pure Appl. Logic}}
  \bibinfo{volume}{138}(\bibinfo{number}{1-3}), pp. \bibinfo{pages}{126--146},
  \doi{10.1016/j.apal.2005.06.009}.

\bibitemdeclare{inproceedings}{LZ20}
\bibitem{LZ20}
\bibinfo{author}{K.~\surnamestart Lehtinen\surnameend} \&
  \bibinfo{author}{M.~\surnamestart Zimmermann\surnameend}
  (\bibinfo{year}{2020}): \emph{\bibinfo{title}{Good-for-games
  {$\omega$}-Pushdown Automata}}.
\newblock In: {\sl \bibinfo{booktitle}{Proc.\ 35th IEEE Symp. on Logic in
  Computer Science}}, pp. \bibinfo{pages}{689--702},
  \doi{10.1145/3373718.3394737}.

\bibitemdeclare{article}{LKH17}
\bibitem{LKH17}
\bibinfo{author}{W.~\surnamestart Li\surnameend}, \bibinfo{author}{Sh.
  \surnamestart Kan\surnameend} \& \bibinfo{author}{Z.~\surnamestart
  Huang\surnameend} (\bibinfo{year}{2017}): \emph{\bibinfo{title}{A Better
  Translation From {LTL} to Transition-Based Generalized B{\"{u}}chi
  Automata}}.
\newblock {\sl \bibinfo{journal}{{IEEE} Access}} \bibinfo{volume}{5}, pp.
  \bibinfo{pages}{27081--27090}, \doi{10.1109/ACCESS.2017.2773123}.

\bibitemdeclare{unpublished}{Mor03}
\bibitem{Mor03}
\bibinfo{author}{G.~\surnamestart Morgenstern\surnameend}
  (\bibinfo{year}{2003}): \emph{\bibinfo{title}{Expressiveness results at the
  bottom of the $\omega$-regular hierarchy}}.
\newblock \bibinfo{note}{{M.Sc.} Thesis, The Hebrew University}.

\bibitemdeclare{techreport}{Myh57}
\bibitem{Myh57}
\bibinfo{author}{J.~\surnamestart Myhill\surnameend} (\bibinfo{year}{1957}):
  \emph{\bibinfo{title}{Finite automata and the representation of events}}.
\newblock \bibinfo{type}{Technical Report} \bibinfo{number}{WADD TR-57-624,
  pages 112--137}, \bibinfo{institution}{Wright Patterson AFB},
  \bibinfo{address}{Ohio}.

\bibitemdeclare{article}{Ner58}
\bibitem{Ner58}
\bibinfo{author}{A.~\surnamestart Nerode\surnameend} (\bibinfo{year}{1958}):
  \emph{\bibinfo{title}{Linear Automaton Transformations}}.
\newblock {\sl \bibinfo{journal}{Proceedings of the American Mathematical
  Society}} \bibinfo{volume}{9}(\bibinfo{number}{4}), pp.
  \bibinfo{pages}{541--544}, \doi{10.2307/2033204}.

\bibitemdeclare{inproceedings}{NW98}
\bibitem{NW98}
\bibinfo{author}{D.~\surnamestart Niwinski\surnameend} \&
  \bibinfo{author}{I.~\surnamestart Walukiewicz\surnameend}
  (\bibinfo{year}{1998}): \emph{\bibinfo{title}{Relating hierarchies of word
  and tree automata}}.
\newblock In: {\sl \bibinfo{booktitle}{Proc. 15th Symp. on Theoretical Aspects
  of Computer Science}}, {\sl \bibinfo{series}{Lecture Notes in Computer
  Science}} \bibinfo{volume}{1373}, \bibinfo{publisher}{Springer}, pp.
  \bibinfo{pages}{320--331}, \doi{10.1007/BFb0028571}.

\bibitemdeclare{inproceedings}{Sch10}
\bibitem{Sch10}
\bibinfo{author}{S.~\surnamestart Schewe\surnameend} (\bibinfo{year}{2010}):
  \emph{\bibinfo{title}{{Beyond Hyper-Minimisation---Minimising DBAs and DPAs
  is NP-Complete}}}.
\newblock In: {\sl \bibinfo{booktitle}{Proc. 30th Conf. on Foundations of
  Software Technology and Theoretical Computer Science}}, {\sl
  \bibinfo{series}{Leibniz International Proceedings in Informatics
  (LIPIcs)}}~\bibinfo{volume}{8}, pp. \bibinfo{pages}{400--411},
  \doi{10.4230/LIPIcs.FSTTCS.2010.400}.

\bibitemdeclare{article}{Sch20}
\bibitem{Sch20}
\bibinfo{author}{S.~\surnamestart Schewe\surnameend} (\bibinfo{year}{2020}):
  \emph{\bibinfo{title}{Minimising Good-for-Games automata is {NP} complete}}.
\newblock {\sl \bibinfo{journal}{CoRR}} \bibinfo{volume}{abs/2003.11979}.

\bibitemdeclare{inproceedings}{SEJK16}
\bibitem{SEJK16}
\bibinfo{author}{S.~\surnamestart Sickert\surnameend},
  \bibinfo{author}{J.~\surnamestart Esparza\surnameend},
  \bibinfo{author}{S.~\surnamestart Jaax\surnameend} \&
  \bibinfo{author}{J.~\surnamestart K{\v r}et{\'{\i}}nsk{\'{y}}\surnameend}
  (\bibinfo{year}{2016}): \emph{\bibinfo{title}{Limit-Deterministic B{\"{u}}chi
  Automata for Linear Temporal Logic}}.
\newblock In: {\sl \bibinfo{booktitle}{Proc. 28th Int. Conf. on Computer Aided
  Verification}}, {\sl \bibinfo{series}{Lecture Notes in Computer Science}}
  \bibinfo{volume}{9780}, \bibinfo{publisher}{Springer}, pp.
  \bibinfo{pages}{312--332}, \doi{10.1007/978-3-319-41540-6_17}.

\bibitemdeclare{article}{Tar72}
\bibitem{Tar72}
\bibinfo{author}{R.E. \surnamestart Tarjan\surnameend} (\bibinfo{year}{1972}):
  \emph{\bibinfo{title}{Depth first search and linear graph algorithms}}.
\newblock {\sl \bibinfo{journal}{SIAM Journal of Computing}}
  \bibinfo{volume}{1(2)}, pp. \bibinfo{pages}{146--160}, \doi{10.1137/0201010}.

\bibitemdeclare{article}{VW94}
\bibitem{VW94}
\bibinfo{author}{M.Y. \surnamestart Vardi\surnameend} \&
  \bibinfo{author}{P.~\surnamestart Wolper\surnameend} (\bibinfo{year}{1994}):
  \emph{\bibinfo{title}{Reasoning about Infinite Computations}}.
\newblock {\sl \bibinfo{journal}{Information and Computation}}
  \bibinfo{volume}{115}(\bibinfo{number}{1}), pp. \bibinfo{pages}{1--37},
  \doi{10.1006/inco.1994.1092}.

\end{thebibliography}

\end{document}